\documentclass[12pt]{article}
\usepackage{graphicx,xspace}
\usepackage{macros}

\title{Online versus Offline Adversaries in Property Testing}

\author{{Esty Kelman\thanks{Boston University, Boston, MA, USA, and Massachusetts Institute of Technology, Cambridge, MA, USA.  Supported in part by the National Science Foundation under Grant No. 2022446 and in part by NSF TRIPODS program (award DMS-2022448). Email: \texttt{ekelman@mit.edu}.}} \and {Ephraim Linder\thanks{Boston University, Boston, MA, USA. Email: \texttt{ejlinder@bu.edu}}}\and {Sofya Raskhodnikova\thanks{Boston University, Boston, MA, USA. Email: \texttt{sofya@bu.edu}}}}
\date{}

\begin{document}
\maketitle
\begin{abstract}
We study property testing with incomplete or noisy inputs. The models we consider allow for \emph{adversarial} manipulation of the input, but differ in whether the manipulation can be done only \emph{offline}, i.e., before the execution of the algorithm, or \emph{online}, i.e., as the algorithm runs. The manipulations by an adversary can come in the form of erasures or corruptions. We compare the query complexity and the randomness complexity of property testing in the offline and online models. 
Kalemaj, Raskhodnikova, and Varma (Theory Comput.\ `23) provide properties that can be tested with a small number of queries with offline erasures, but cannot be tested at all with online erasures. We demonstrate that the two models are incomparable in terms of query complexity: we construct properties that can be tested with a constant number of queries in the online corruption model, but require querying a significant fraction of the input in the offline erasure model. We also construct properties that exhibit a strong separation between the randomness complexity of testing in the presence of offline and online adversaries: testing these properties in the online model requires exponentially more random bits than in the offline model, even when they are tested with nearly the same number of queries in both models. Our randomness separation relies on a novel reduction from randomness-efficient testers in the adversarial online model to query-efficient testers in the standard model.
\end{abstract}

\thispagestyle{empty}
\newpage

\pagenumbering{arabic} 
\section{Introduction}\label{sec:introduction}
Property testing~\cite{RubinfeldS96,GGR98} aims to lay algorithmic foundations for processing big data. It is a formal study of fast algorithms that accept objects with a given property and reject objects that are far. The goal of this work is to compare models of property testing that address the situations when the data that needs to be analyzed is not only large, but also incomplete or noisy. The models we consider allow for adversarial manipulation of the input, but differ in whether the manipulation can be done only {\em offline}, i.e., before the execution of the algorithm, or {\em online}, i.e., as the algorithm runs. The manipulations by an adversary can come in the form of erasures or corruptions.  
The offline erasure-resilient property testing model was proposed by Dixit et al.~\cite{DixitRTV18}. The model that captures offline corruptions is called {\em tolerant testing} and was introduced by Parnas et al.~\cite{ParnasRR06}. The online analogues of these models were recently defined by Kalemaj et al.~\cite{KalemajRV23}. 

We compare the query complexity and the randomness complexity of property testing in the offline and online models. In the offline-erasure model, a constant fraction of the input is erased adversarially before the execution of the algorithm, whereas in the online-erasure model, the adversary is allowed to make a fixed number of erasures after each query.
It may seem that the online adversary has strictly more power, and thus testing in the online models should require more queries. Indeed, \cite{KalemajRV23} justify this intuition and show that simple properties---sortedness and the Lipschitz property of sequences---cannot be tested at all with online erasures\footnote{In contrast, in the offline models, testing is always possible when the entire input is read.}, even though they are testable using a small number of queries when erasures are offline.
Intuitively, \cite{KalemajRV23} show that testing in the presence of one online erasure per query (performed by an adversary with the knowledge of previous queries) can be harder than testing when a constant fraction of the input is adversarially erased offline (in advance).
The following natural open question is stated in \cite{KalemajRV23}: ``Is there a property that
has smaller query complexity in the online-erasure-resilient model than in the (offline)
erasure-resilient model of \cite{DixitRTV18}?''
We answer this question in the affirmative: specifically, we construct a property that can be tested with a constant number of queries in the online-erasure model (and even in the (harder) online-corruption model), but requires a nearly linear number of queries in the offline-erasure model.   
We conclude that the two adversarial erasure models are incomparable.

Our second result concerns the randomness complexity of testing in the online and offline manipulation models. Randomness is an important tool in the design of a wide variety of algorithms, and extensive research has been conducted to understand the power of randomness in computation. 
The investigation of the randomness complexity of property testing algorithms was initiated by Goldreich and Sheffet \cite{GS10}. They showed that all property testers in the standard model can be derandomized to a large extent: they can be converted to testers that use the number of random bits that is logarithmic in the size of the input, while incurring at most a constant factor blowup in query complexity. The result of \cite{GS10} easily extends to all offline testing models, but does not apply in the online setting. Moreover, the existing testers in the online model \cite{KalemajRV23,MinzerZ24,Ben-EliezerKMR24,Westover2024,Arora2025} use lots of randomness to fool the adversary. A natural question that arises is how much derandomization is possible in the online setting.
We show that a lot of randomness is sometimes indeed necessary to test in the presence of an online adversary. In particular, we exhibit a property that requires exponentially more random bits to test in the online manipulation model than in the offline manipulation model. This result provides a natural converse to the statement of \cite{MinzerZ24} that ``the (online) adversary cannot foil our plan if there is no plan" --- i.e., random queries are robust to online manipulation of the input.

\subsection{Our results}
We represent the input to our property testing algorithm as a string. The algorithm is given access to its input string via queries (specifically, of the form ``what is the character at index $i$?"). For each proximity parameter $\eps\in(0,1)$, the algorithm has to distinguish between strings that have a specified property and strings that are $\eps$-far from having the property --- i.e., strings that need to be changed on at least an $\eps$-fraction of indices in order to satisfy the property. The query complexity (the number of queries to the input) of the tester is measured in terms of $\eps$ and the input length $n$. In this section, for simplicity, $\eps$ is considered to be a constant and omitted from the statements.

\subsubsection{Query complexity: Model Incomparability}

We show that the online and offline adversarial models of property testing are incomparable in terms of the query complexity.
Recall that \cite{KalemajRV23} prove that sortedness and Lipschitzness of sequences are efficiently testable with offline erasures, but cannot be tested at all even with one online erasure per query\footnote{We note that there is an even simpler property than sortedness and Lipschitzness of sequences that is easily testable with offline erasures or corruptions, but is not testable at all with online erasures. Moreover, this property is over the alphabet $\Sigma=\{0,1\}$. Specifically, $\cP=\{ww: w\in \{0,1\}^*\}$. We can estimate the distance to this property within additive error $\eps$ using $O(\frac 1{\eps^2})$ queries (and therefore it is easily tolerantly testable). However, we cannot test it with online erasures because, whenever the tester queries some position $p$ in an input string of length $2n$, the adversary can erase the corresponding position ($n+p$ if $p\leq n$ and $p-n$ if $p>n$).\label{footnote:ww-reverse}}.
 We complement this result with the following theorems, which show that testing in the presence of a small fraction of offline erasures can be harder than testing with online erasures. The first of this theorems (\Cref{thm:query separation intro version - general -informal 1}) considers the regime with many online erasures per query; the second  (\Cref{thm:query separation intro version - general -informal 1}) considers the regime with one erasure per query. The former exhibits a logarithmic gap in query complexity, where as the latter exhibits a nearly linear gap.
\begin{theorem}[Informal version of \Cref{cor:query separation cor 1}]\label{thm:query separation intro version - general -informal 1}
    There exists a property that can be tested using a \emph{constant number of queries} in the presence of an online adversary that makes $O\big(\frac n{\log n}\big)$ erasures per query; but requires \emph{$\widetilde{\Omega}(\log n)$ queries}\footnote{The notation 
    $\widetilde{\Omega}(f(n))$ hides $\polylog f(n)$ factors, e.g., $\widetilde{\Omega}(\log n)=\Omega(\log n/\polyloglog n)$.} to test when the characters at $\Theta\big(\frac n{\log\log n}\big)$ indices are erased offline (in advance).
\end{theorem}
We cannot expect to prove an analogue of \Cref{thm:query separation intro version - general -informal 1} for the case when the number of online erasures per query is the same as the total number of offline erasures. This is because, after the first query, an online adversary can erase the same part of the input that the offline adversary would have erased in advance.
\begin{remark}
    Although the online erasure rate in \Cref{thm:query separation intro version - general -informal 1} is very large, the total number of erasures available to the online adversary is smaller than the number of offline erasures we considered. It is an intriguing open question whether there exists a property for which every tester in the online model accumulates at least as many erasures as the best tester in the offline model, but also requires fewer queries to test in the online model than in the offline model.
\end{remark}
At the other end of the spectrum, we show that if the number of erasures is much smaller in the online model than in the offline model, then there is an exponentially large gap---i.e., there are properties that can be tested with a constant number of queries when the erasures are online but require a nearly linear number of queries in the offline erasure model. 
\begin{theorem}[Informal version of \Cref{cor:query separation cor 2}]
\label{thm:query separation intro version - general -informal 2}
     There exists a property that can be tested with a constant number of queries in the presence of an online adversary that makes one erasure per query, but requires $\widetilde{\Omega}(n)$ queries to test when the characters at $\Theta\big(\frac n{\log n}\big)$ indices are erased offline (in advance).
\end{theorem}

In fact, the query complexity separations in \Cref{thm:query separation intro version - general -informal 1,thm:query separation intro version - general -informal 2} hold even if the online adversary can make {\em corruptions} instead of erasures.

To prove \Cref{thm:query separation intro version - general -informal 1,thm:query separation intro version - general -informal 2}, we show that any property that separates the offline model from the standard one, i.e., a property that is easy to test in the standard model but hard in the offline model, can be ``lifted" to separate the offline and online models. We can then leverage the state of the art separation (between the standard and offline-manipulation testing models) of \cite{BenFLR20} to obtain strong separations between the online and offline manipulation models.

Our ``lifting" operation is simply encoding the property with a repetition code. Intuitively, repetition won't make the property easier to test in the presence of offline erasures since the same erasures can be made on all copies. 
Our main technical proof for this part is to show that repetition code is robust against online adversaries. We prove the following general result.
\begin{lemma}[Informal version of \Cref{lem: lifting}]
    If a property $\cP$ is testable when no adversary is present then the repeated version $\cP^r$, given by concatenating each string in $\cP$ with itself $r$ times, is testable in the presence of an online adversary (as long as the number of erasures per query is not too large).
\end{lemma}

\subsubsection{Randomness complexity separations}
As discussed earlier, the investigation of
the randomness complexity of property testing algorithms was initiated by Goldreich and Sheffet \cite{GS10} who showed that every property that is testable in the standard model using $q$ queries can be tested using $O(q)$ queries and $O(\log n)$ random bits. We build on their investigation by studying the randomness complexity of the erasure models. Though their result easily extends to the offline models, the same extension does not work in the online model. Indeed, we show that the derandomization of \cite{GS10} cannot be extended to the online testing model---that is, we construct a property  that requires $\omega(\log n)$ random bits to test in the online model. 

Note that a randomness separation trivially holds for properties which are not testable with online erasures, or when the offline-erasure-resilient tester is allowed to query the entire input\footnote{In this case, the tester in the offline model is deterministic. Note that it is impossible to deterministically
test any nontrivial property in the online model, since the adversary can erase all but the first query made by a deterministic tester.}.
Prior to our work it was not clear if there is a randomness separation that applies to properties that are testable in the online model, and still holds if we require that the testers in both models to have a sublinear (in $n$) query complexity. 

We demonstrate such a separation. In particular, we show that there exists a property that is testable in both models using a sublinear number of queries, 
and, moreover, testing this property in the online erasure model requires exponentially more random bits than testing it in the offline erasure model. 

\begin{theorem}[Informal version of \Cref{cor:rand sep cor 1}]\label{thm:randomness separation intro version -informal}
There exists a property that can be tested using $O(\sqrt{n})$ queries in the presence of either an online or an offline adversary.
The tester against the offline adversary uses $O(\log n)$ random bits, whereas every tester in the presence of an online adversary that makes one erasure per query requires $\widetilde{\Omega}(\sqrt{n})$ random bits to succeed with constant probability. 
\end{theorem}

\Cref{thm:randomness separation intro version -informal} follows easily from a more general result, \Cref{thm: rand sep}, which we prove in \Cref{sec:randomness-separation} and which yields meaningful separations for larger online erasure rates as well. The main technique we develop to prove a lower bound on the randomness complexity of testing with an online adversary is a transformation from randomness-efficient testers in the online model to query efficient testers in the standard model. 
A special case of the guarantees of our transformation is stated in \Cref{lem: rand-query reduction intro}. (The general version appears in \Cref{lem: rand-query reduction}.)

To prove \Cref{thm: rand sep}, we combine our transformation with existing results regarding the property of strings called $\tau$-Distinct-Elements, which is parametrized by $\tau\in\N$, and consists of strings that have at most $\tau$ distinct characters. Testing $\tau$-Distinct-Elements was recently investigated in \cite{GoldreichR23,AdarFL24,AdarF24,PintoH24}, and it is a natural testing version of the distinct elements approximation problem studied in \cite{RaskRSS07,ValiantV10}. For a more detailed exposition regarding the history of the $\tau$-Distinct-Elements property and its variants see \cite{PintoH24}.

\begin{lemma}[Special case of \Cref{lem: rand-query reduction} for one online erasure per query]
    \label{lem: rand-query reduction intro}
    Fix $r\in\N$.  If a property  $\cP$ can be tested using $r$ random bits in the presence of an online adversary that makes one erasure per query, then $\cP$ can be tested in the standard model (with the same proximity parameter $\eps$) using at most $r$ queries. 
\end{lemma}
The technique from \Cref{lem: rand-query reduction intro}, combined with the derandomization results of \cite{GS10}, immediately yields a separation for $\eps=o\big(\frac{1}{\log n}\big)$: 
Consider the property $\cP\subset\zo^n$ of being the zero string. By the folklore property testing bound, $\Theta(\frac 1\eps)$ queries are necessary and sufficient to test $\cP$. By \cite{GS10}, there exists a tester in the offline erasure model that uses $O(\log n)$ random bits. However, by \Cref{lem: rand-query reduction intro}, every tester in the online model must use $\Omega(\frac 1\eps)=\omega(\log n)$ random bits. While this works for small $\eps$, the separation we show in \Cref{thm:randomness separation intro version -informal} (and the more general \Cref{thm: rand sep}) holds for constant $\eps$ as well. 

\subsection{Related work}\label{sec:relationship}
\paragraph{Relationship between offline and standard models}
In terms of query complexity, testing with offline erasures lies between standard and tolerant testing. By definition, an offline-erasure-resilient tester is also a property tester in the standard model, because it has to work on the inputs with no erasures. 
Not surprisingly, offline testing with erasures is no harder than offline testing with corruptions. This is formalized in \cite[Theorem 1.4]{DixitRTV18} that shows that
the existence of a tolerant tester
for a property implies the existence of an offline-erasure-resilient tester for that property for a comparable setting of parameters. There are also separations that show that standard testing is strictly easier than offline-erasure-resilient testing, which in turn is strictly easier than tolerant testing. Specifically,  Dixit et
al.~\cite{DixitRTV18} showed that a property defined by 
Fischer and Fortnow \cite{FischerF06} can be tested in the standard model with a constant number of queries, but requires $n^{\Omega(1)}$ queries in the offline-erasure-resilient model. This separation was strengthened by Ben-Eliezer et al.~\cite{BenFLR20}, improving $n^{\Omega(1)}$ to $\widetilde{\Omega}(n)$. Finally, Raskhodnikova et al.~\cite{RRV21} prove a separation similar to that of \cite{DixitRTV18} between offline-erasure-resilient testing and tolerant testing. Thus, at a high level, we have a complete picture in terms of the relative difficulty of testing in the three offline models.
\paragraph{Relationship between offline and online testing}
Kalemaj et al.~\cite{KalemajRV23} exhibit two properties of strings for which testing with online erasures is impossible for every proximity parameter $\eps\leq  \frac 1 2$, whereas testing in the offline models is easy. 
The first property, {\em sortedness}, consists of sorted arrays, i.e., strings $x$ of length $n$ such that $x_i\leq x_{i+1}$ for all $i\in[n-1]$. 
The query complexity of testing sortedness (for constant $\eps$) is $\Theta(\log n)$
in both the standard and the offline erasures model (when the erased part is at most a constant fraction of the input) \cite{EKKRV,DGLRRS99,Ras99,Fischer04,BGJRW12,ChSe14,Enc1,Belovs18,DixitRTV18}.
The second property, Lipschitzness, consists of  strings 
$x\in \{0,1,2\}^n$ satisfying $|x_i- x_{i+1}|\leq 1$ for all $i\in[n-1]$. It can be tested with $\Theta(\frac 1\eps)$ queries in both the standard and the offline models \cite{JhaR13,DixitRTV18}.

\section{Preliminaries}\label{sec:model-definitions}

\paragraph{Notation} We use $[n]$ to denote the set of integers $\{1,2,\dots, n\}$, and $\N$ to denote the set of positive integers.

\paragraph{Property testing}
We start by stating standard property testing definitions. We represent an input to a property testing algorithm as a string of length $n$ over some alphabet $\Sigma_n$ that might depend on $n$. For example, $\Sigma_n$ might be $\{0,1\}$ or $[n]$.

\begin{definition}[Relative Hamming distance, property, $\eps$-far]
\label{def:property}
The \emph{relative Hamming distance} between two strings $x,y$ of length $n$ is $\delta_H(x,y)=\Pr_{i\sim [n]}[x_i\neq y_i]$ where $i$ is uniformly random index from $[n]$. For a set $\cS$ of strings of length $n$, define $\delta_H(x,\cS)=\min_{y\in \cS} \delta_H(x,y)$. A property $\cP$ is a subset of $\Sigma^*$ given by $\bigcup_{n\in\N} \cP_n$, where each $\cP_n$ consists of strings of length $n$ over some alphabet $\Sigma_n$. 
A string $x$ of length $n$ is \emph{$\eps$-far} from $\cP$ if $\delta_H(x,\cP_n)\geq \eps$.
\end{definition}
\begin{definition}[$\eps$-tester in the standard model \cite{RubinfeldS96,GGR98}]\label{def:property-testing}
    
    For every property $\cP\subseteq \Sigma^*$ and proximity parameter $\eps\in(0,1)$, an algorithm $\cT$ is an \emph{$\eps$-tester} for $\cP$ if, given a parameter $n\in \N$ and oracle access to input $x\in\Sigma^n$, the algorithm $\cT$ accepts with probability at least $\frac 23$ whenever $x\in\cP$ and rejects with probability at least $\frac 23$ whenever $x$ is $\eps$-far from $\cP$. A tester has a one-sided error if it always accepts inputs $x\in\cP$.
\end{definition}
\paragraph{Offline manipulations}
The offline model with erasures was introduced by \cite{DixitRTV18} and subsequently studied in~\cite{RV18,RRV21,BenFLR20,LeviPRV22}. The definition we use here is adapted from \cite{LeviPRV22} and only differs from the original definition in how the parameter $\eps$ is interpreted. We use $\perp$ to denote an erased symbol in the input.

\begin{definition}[$\alpha$-erased string, completion]
For each $\alpha\in (0,1)$, a string $x\in (\Sigma\cup \{\perp\})^n$ is $\alpha$-erased if at most an $\alpha$ fraction of symbols in it are $\perp$.  The indices of the $\perp$ symbols in the string are called {\em erased}. 
A string $y\in\Sigma^n$ that differs from a string $x\in(\Sigma\cup\{\perp\})^n$ only on indices erased in $x$ is called a {\em completion} of~$x$.
\end{definition}
\begin{definition}[Offline-erasure-resilient tester \cite{DixitRTV18}]\label{def:erasure-resilient-tester}
For every property $\cP\subset \Sigma^*$ and parameters $\alpha\in[0,1)$ and $\eps\in(0,1)$, an algorithm $\cT$ is an {\em $\alpha$-offline-erasure-resilient $\eps$-tester} for $\cP$ if, given a parameter $n\in\N$ and oracle access to an $\alpha$-erased input $x\in\Sigma^n$, the algorithm $\cT$ 
 accepts with probability at least $\frac 23$ whenever there exists a completion $y\in\cP$ of $x$ and 
 rejects with probability at least $\frac 23$ whenever every completion $y$ of $x$ is $\eps$-far from~$\cP$.

\end{definition}
When $\alpha=0$, the $\alpha$-offline-erasure-resilient model is the same as the standard model. 
Another generalization of property testing is tolerant testing, introduced by \cite{ParnasRR06} and extensively studied over the last decades.
It can be viewed as guaranteeing resilience to an $\alpha$ fraction of the input being corrupted by an offline adversary.

\begin{definition}[$(\alpha,\eps)$-tolerant tester \cite{ParnasRR06}]
    For every property $\cP\subset \Sigma^*$ and parameters $\alpha,\eps\in(0,1)$, an algorithm $\cT$ is an \emph{$(\alpha,\eps)$-tolerant tester} for $\cP$ if, given a parameter $n\in\N$ and oracle access to input $x\in\Sigma^n$, the algorithm $\cT$ accepts with probability at least $\frac 23$ whenever $x$ is $\alpha$-close to $\cP$ and rejects with probability at least $\frac 23$ whenever $x$ is $\eps$-far from~$\cP$.
\end{definition}
When we want to stress comparison to other models we study, we call an $(\alpha,\eps)$-tolerant tester an {\em $\alpha$-offline-corruption-resilient $\eps$-tester}.

\paragraph{Online manipulations}
In this model, the input string $x\in\Sigma^n$ is accessed via an adversarial oracle $\cO$. After answering each query made by the algorithm, the adversary can erase or corrupt a small number of data points. At the beginning of the execution of the algorithm, $\cO(i)=x_i$ for all $i\in[n]$. So, the first query is always answered correctly. The number of queries the adversary can manipulate {\em after} answering each query is parameterized by $\rate\in \N$. The manipulated values are used by the oracle to answer future queries to the corresponding indices. The algorithm does not know which input locations have been tempered with.
As stated in \cite{KalemajRV23}, ``the actions of the oracle can depend on the input, the queries made so far, and even on the publicly known code that the algorithm is running, but {\em not} on future coin tosses of the algorithm.''

\begin{definition}[Adversarial oracle, online testers \cite{KalemajRV23}]
Fix $\rate\in \N$. A {\em $\rate$-online-erasure} oracle can replace values $\cO(i)$ on up to $\rate$ indices $i\in[n]$ with the erasure symbol $\perp$ after answering each query. A {\em $\rate$-online-corruption} oracle is defined analogously, except that it replaces each symbol with an arbitrary symbol from $\Sigma$ instead of erasing it.
For each $\rate\in \N$, a {\em $\rate$-online-erasure-resilient $\eps$-tester} $\cT$ is defined as in the standard model (\Cref{def:property-testing}), except that it accesses its input via a $\rate$-online-erasure oracle as opposed to querying the input directly. The {\em $\rate$-online-corruption-resilient $\eps$-tester} is defined analogously.
\end{definition}

The standard property testing model is a special case of this enhanced model and corresponds to the case when $\rate=0$. Moreover, it is clear from the definition that testing with online erasures is no harder than testing with online corruptions. Specifically, every $\rate$-online-corruption-resilient tester can be simulated by a tester that accesses its input via a $\rate$-online-erasure oracle---the tester can simply replace each $\perp$ with an arbitrary value from~$\Sigma$.

\section{Query complexity separation}
\label{sec:query-separation}

In this section, we prove that testing with an offline adversary can require more queries than testing with an online adversary.

\begin{theorem}[Query complexity separation]
    \label{thm: query sep}
    Fix  $\ell\in\N$ and let $r=r(n)<n$ be a function of the input length $n$. 
    There exists a property $\cP\subseteq\zo^*$ and a constant $\eps_1=\eps_1(\ell)\in(0,1)$ such that:\footnote{We use  $\log^{(\ell)}$ to denote the $\log$ function applied $\ell$ times.} 
    \begin{enumerate}
        \item\label{thm: query sep item 1} For all $\eps\in(0,1)$ and $\rate=(\frac \eps 2)^{O(\ell)}\cdot r$, the property $\cP$ is $\rate$-online-corruption-resiliently $\eps$-testable using $(\frac 2 \eps)^{O(\ell)}$ queries. 
        \item\label{thm: query sep item 2} For all $n\in\N$, $\eps\in(0,\eps_1),$ and $\alpha=\Omega\left(\frac{1}{\log^{(\ell)}(n/r)}\right)$ such that $\eps+\alpha<1$, every $\alpha$-offline-erasure-resilient $\eps$-tester for $\cP$ must make $\Omega\left(\frac{n}{r\cdot\polylog^{(\ell)}(n/r)}\right)$ queries on inputs of length~$n$. 
    \end{enumerate}
\end{theorem}

For $\ell=1$ and $r=\frac n{\log n}$, we obtain the following corollary.
\begin{corollary}\label{cor:query separation cor 1}
There exist a property $\cP$ such that for all $\eps\in(0,1)$ and $\rate\le \poly(\eps)\cdot \frac{\plen}{\log \plen}$, the property $\cP$ is $\eps$-testable using $\poly(\frac 1 \eps)$ queries in the $\rate$-online-corruption model. However, for all $\alpha=\Omega\left(\frac{1}{\log\log \plen}\right)$ such that $\alpha+\eps<1$, every $\alpha$-offline-erasure-resilient $\eps$-tester requires $\widetilde{\Omega}(\log \plen)$ queries.    
\end{corollary}
While \Cref{cor:query separation cor 1} shows that testing in the $\alpha$-offline-erasure model can be harder than testing in the $\rate$-online-model for large $\rate$ and small $\alpha$, the difference in query complexity is quite mild. To obtain a 
large gap in query complexity from \Cref{thm: query sep}, we fix $\ell\in\N$ and constant proximity parameter $\eps\le\eps_1(\ell)$, and $r=r(\eps,\ell)$ such that $\rate=1$. This yields the following corollary.

\begin{corollary}
\label{cor:query separation cor 2}
Fix $\ell\in \N$ and a constant $\eps=\eps(\ell)$. There exists a property $\cP$ such that $\cP$ is $\eps$-testable in the $1$-online-corruption model using constantly many queries. But, for all $\alpha=\Omega\left(\frac{1}{\log^{(\ell)} \plen}\right)$ such that $\alpha+\eps<1$, every $\alpha$-offline-erasure-resilient $\eps$-tester requires $\Omega(\plen/\polylog^{(\ell)} \plen)$ queries.    
\end{corollary}

One feature of the results in \Cref{thm: query sep} (as well as \Cref{cor:query separation cor 1,cor:query separation cor 2})  is that they hold even when the adversary in the online model is allowed to make corruptions, while the adversary in the offline model is restricted to erasures. 

To prove \Cref{thm: query sep}, we introduce the following ``lifting" result (\Cref{lem: lifting}). Informally, the lemma states that every property $\cP$ that is testable in the standard model has an encoding $\cP'$ that is testable with the same query complexity even in the presence of an online-corruption adversary. This result allows us to transfer existing separations between the standard model and the offline-erasure model to a separation between the online-corruption model and the offline-erasure model. 

To ``lift'' a property $\cP$, we simply encoded it with a repetition code, that is, repeat the corresponding input string.
\begin{definition}[$r$-concatenated string $x^r$ and property $\cP^r$]
\label{def:r-string}
    For all $r\in\N$ and strings $x$, let $x^r$ denote the concatenation of $r$ copies of $x$, written $x^r[1]\circ x^r[2]\circ\dots\circ x^r[r]$. Additionally, for a property $\cP$ of strings, let $\cP^r$ denote the property $\{x^r : x\in \cP\}$. 
\end{definition}

Our lifting lemma holds for properties of strings over any alphabet.

\begin{lemma}[Lifting lemma]
    \label{lem: lifting}
    Let $\Sigma$ be an alphabet and $\delta\in(0,1)$ be a sufficiently small constant. Let $\cP\subseteq\Sigma^*$ be a property of strings
    that is $\eps$-testable in the standard model using $q(m,\eps)$ queries on inputs of length m, 
    where $q(m,\eps)=\Omega\big(\frac1\eps\big)$.
    Then, for all $r\in\N$ and $\rate\leq\frac{\delta\cdot r}{[q(m,\frac \eps 2)\log q(m,\frac \eps 2)]^2}$, the property $\cP^r$ is $\rate$-online-corruption-resiliently $\eps$-testable using $\widetilde{O}(q(m,\frac \eps 2))$ queries on inputs of length $n=m\cdot r$.
\end{lemma}

Next we use \Cref{lem: lifting} to prove \Cref{thm: query sep}, deferring the proof of \Cref{lem: lifting}
to \Cref{sec: lifting proof}.
We leverage the following result of Ben-Eliezer, Fisher, Levi, and Rothblum \cite{BenFLR20}, which separates the offline-erasure-resilient model from the standard model. 
\begin{theorem}[\protect{\cite[Theorem 6.2]{BenFLR20}}]
    \label{thm: stanadard offline sep}
    For all constant $\ell\in\N$, there exist a property $\mathcal{Q}^{(\ell)}\subseteq \{0,1\}^*$
    and a constant $\eps_1=\eps_1(\ell)\in(0,1)$ such that the following hold:
    \begin{enumerate}
        \item For every $\eps\in(0,1)$, the property $\mathcal{Q}^{(\ell)}$ can be $\eps$-tested using $(\frac2\eps)^{O(\ell)}$ queries.
        \item For all $m\in\N$, $\eps\in (0,\eps_1),$ and $\alpha=\Omega(1/\log^{(\ell)}m)$ satisfying $\eps+\alpha<1$, every $\alpha$-erasure resilient $\eps$-tester for $\mathcal{Q}^{(\ell)}$ must make $\Omega(m/(10^\ell\cdot \polylog^{(\ell)}m))$ queries on inputs of length $m$.
    \end{enumerate}
\end{theorem}

\begin{proof}[Proof of \Cref{thm: query sep}]
    Fix $\ell\in\N$ and let $m,r\in\N$ be sufficiently large. Let $\cP$ denote the property $\mathcal{Q}^{(\ell)}$ of strings of length $m$ given by \Cref{thm: stanadard offline sep}. By \Cref{thm: stanadard offline sep}, the property $\cP$ can be $\eps$-tested using $\left(\frac2\eps\right)^{O(\ell)}$ queries. Thus, \Cref{lem: lifting} guarantees that for all $\rate=(\frac \eps 2)^{O(\ell)}\cdot r$, the property $\cP^r$ is $\rate$-online-corruption-resiliently $\eps$-testable using $\left(\frac2\eps\right)^{O(\ell)}$ queries on inputs of length $n=mr$. 
    
    Now we  prove that $\cP^r$ requires $\Omega\left(\frac{n}{r\polylog^{(\ell)}(n/r)}\right)$ queries to $\eps$-test in the $\alpha$-offline-erasure-resilient model. To do that, we show how to use any $\alpha$-offline-erasure-resilient $\eps$-tester $\cT_r$ for $\cP^r$ to construct an $\alpha$-offline-erasure-resilient $\eps$-tester $\cT$ for $\cP$ with the same query complexity as $\cT_r$. Given a string $x^r$, let $x^r[\rindex]_\mindex$ denote index $\mindex$ of the $\rindex$-th copy of $x$. Observe that for every $\alpha$-erased input $x$, the string $x^r$ is also $\alpha$-erased and the distance of $x$ from $\cP$ is the same as the distance of $x^r$ from $\cP^r$; moreover, for all $\rindex\in[r]$ and $\mindex\in[m]$, we have $x^r[\rindex]_\mindex=x_\mindex$. Thus, the tester $\cT,$ given query access to $x$, can simulate $\cT_r$ executed with query access to $x^r$ and output the result given by $\cT_r$. By \Cref{thm: stanadard offline sep}, the tester $\cT_r$ has query complexity $\Omega\paren{\frac{m}{\polylog^{(\ell)} m}}=\Omega\left(\frac{n}{r\polylog^{(\ell)} (n/r)}\right)$ whenever $\alpha=\Omega\left(\frac1{\log^{(\ell)}(n/r)}\right)$. 
\end{proof}

\subsection{Proof of the lifting lemma (\Cref{lem: lifting})}
\label{sec: lifting proof}

    Recall that \Cref{lem: lifting} states that if a property $\cP$ is testable in the standard model, then the property $\cP^r$ is testable in the online-corruption model. Let $\cT$ be an $\eps$-tester for $\cP$ that has failure probability $\frac{1}{12}$, and query complexity $c_0\cdot q(m,\eps)$ for some constant $c_0>0$ (we can obtain such a tester via standard amplification techniques).  We construct a tester $\cT'$ for $\cP^r$ (\Cref{alg: online test}) and show that it is $\rate$-online-corruption-resilient. The tester $\cT'$  consists of two phases: first, $\cT'(s)$ calls \reptest (\Cref{alg: rep-test}) to test that $s$ is a repetition of some string, $w$, and second, $\cT'(s)$ simulates $\cT$ with query access to $w$. 
    
    For a string $s$ of length $n=r\cdot m$ composed of $r$ substrings each of length $m$, we use the notation $s[\rindex]_\mindex$ to denote index $\mindex$, in the $\rindex$-th substring of $s$.
   
    \begin{algorithm}[H]
\textbf{Parameters:} length parameter $m$, repetition parameter $r$, proximity parameter $\eps\in(0,1)$\\
    \textbf{Input:} query access to string  $s=s[1]\circ\dots\circ s[r]$ such that $|s[\rindex]|=m$ for each $\rindex\in[r]$\\
    \textbf{Subroutines:} \reptest (\Cref{alg: rep-test}) and $c_0\cdot q(m,\eps)$-query tester $\cT$ for $\cP$
	\begin{algorithmic}[1]
		\caption{\label{alg: online test} $\rate$-online-corruption-resilient $\eps$-tester $\cT'$}
        \State $c_1\gets 24\cdot c_0$
        \If{$\reptest(s,\frac \eps 2)$ rejects} \textbf{reject}
        \EndIf  
        \State\label{step:simulation} simulate $\cT$ with proximity parameter $\frac \eps 2$, answering each query $\mindex$ made by $\cT$ as follows:
        
        sample a set of $d=\log (c_1\cdot q(m,\frac \eps 2))$ uniform indices $\rindex_1, \dots, \rindex_d\in [r]$
        
        \textbf{if} there exists a pair of indices $k, k'\in [d]$ such that $s[\rindex_k]_\mindex\neq s[\rindex_{k'}]_\mindex$ \textbf{then} \textbf{reject}
        
        \textbf{else} provide the answer $s[\rindex_1]_\mindex$ to $\cT$
        \If{the simulation of $\cT$ rejects} \textbf{reject}
        \Else\textbf{ accept}
        \EndIf
	\end{algorithmic}
\end{algorithm}

\begin{algorithm}[H]
\textbf{Parameters:} length parameter $m$, repetition parameter $r$, proximity parameter $\eps\in(0,1)$\\
    \textbf{Input:} query access to string $s=s[1]\circ\dots\circ s[r]$ such that $|s[\rindex]|=m$ for each $\rindex\in[r]$
	\begin{algorithmic}[1]
		\caption{\label{alg: rep-test} $\reptest$} \State \textbf{repeat} $\frac 2\eps$ 
        times: 
        
            sample $\rindex_1,\rindex_2\sim[r]$ and $\mindex\sim[m]$ uniformly and independently at random
        \If{$s[\rindex_1]_\mindex\neq s[\rindex_2]_\mindex$} \textbf{reject}\EndIf
        \State\textbf{else accept}
	\end{algorithmic}
\end{algorithm}

We start by analyzing \reptest (\Cref{alg: rep-test}).

    \begin{claim}
    \label{claim:reptest}
    For all alphabets $\Sigma$ and parameters $\eps\in(0,1)$ and $r,m\in\N$, $\reptest$ is a one-sided error $\eps$-tester for the repetition code defined as  $\cC_{m,r}=\{s^{r}: s\in\Sigma^m\}$.
    It makes $O(\frac 1\eps)$ queries and has error probability at most $\frac{1}{6}$. 
    \end{claim}
    \begin{proof}
Let $s=s[1]\circ\dots \circ s[r],$ where $s[\rindex]=s[\rindex]_1\dots s[\rindex]_m$ for all $\rindex\in[r]$, be the input string over the alphabet $\Sigma$.
     If $s\in\cC_{m,r}$, then \reptest always accepts $s$.   
        
        Suppose that string $s$ is $\eps$-far from $\cC_{m,r}$.
        Let $\pl_{\rindex\in[r]}(a_\rindex)$ denote a function that takes inputs $a_1,\dots,a_r\in \Sigma$ and outputs the most frequent among them, resolving ties arbitrarily.
        Let $\hat{w}$ be the string of length $m$ defined by $\hat{w}_\mindex=\pl_{\rindex\in [r]} s[\rindex]_\mindex$ for all $\mindex\in[m]$. Since $s$ is $\eps$-far from $\cC_{m,r}$, we can bound the distance between the input string and $\hat{w}$ repeated $r$ times:
        \begin{equation}\label{eq:distance-from-plurality}
           \delta_H(s,\hat{w}^r)\geq \eps. 
        \end{equation}

For all $\mindex\in[m]$ and $a\in\Sigma$, let $p_\mindex(a)$ denote $\Pr_{\rindex\in[r]}[s[\rindex]_\mindex=a]$, the probability that a randomly chosen repetition of the $\mindex$-th character in $s$ is equal to $a$. The probability that one iteration of the loop in \reptest accepts $s$ is
\begin{align*}
    \Pr_{\substack{ \rindex_1,\rindex_2\in[r] \\ \mindex\in[m] }}\big[s[\rindex_1]_\mindex
    =s[\rindex_2]_\mindex\big]
    &= \frac 1 m \sum_{\mindex\in[m]}\sum_{a\in\Sigma} [p_\mindex(a)]^2 & \text{by law of total probability}\\
    &\leq  \frac 1 m \sum_{\mindex\in[m]} p_\mindex(\hat{w}_\mindex)\sum_{a\in\Sigma} p_\mindex(a) & \text{\ \ \ \ since $p_\mindex(a)\leq p_\mindex(\hat{w}_\mindex)\ \forall \mindex\in[m],a\in\Sigma$} \\
    &=  \frac 1 m \sum_{\mindex\in[m]} p_\mindex(\hat{w}_\mindex) & \text{since $\sum_{a\in\Sigma} p_\mindex(a)=1 \ \forall \mindex\in[m]$}\\
    &= \Pr_{\substack{\rindex\in[r] \\ \mindex\in[m]}}[s[\rindex]_\mindex=\hat{w}_\mindex]& \text{by law of total probability}\\
    &=1-\delta_H(s,\hat{w}^r) &\text{by definition of $\delta_H$}\\
    &\leq 1-\eps&\text{by \eqref{eq:distance-from-plurality}.}
\end{align*}
 Hence, each iteration of the loop in \reptest rejects with probability at least $\eps$. Therefore, $\reptest$ accepts $s$ with probability at most $\left(1-\eps\right)^{2/\eps}\leq e^{-\eps\cdot 2/\eps }\leq \frac{1}{6}$. 
  \end{proof}

    Next, we analyze the tester $\cT'$.

    \begin{claim}[$\cT'$ is a tester in the standard model] \label{claim: tester with no adversary}
    For all $\eps\in(0,1)$,$r\in\N$, alphabets $\Sigma$, and properties $\cP\in\Sigma^*$, \Cref{alg: online test} is an  $\eps$-tester for $\cP^r$ in the standard model (without any adversary). Moreover, it accepts $\eps$-far inputs with probability at most $\frac{1}{6}$, and has one-sided error whenever $\cT$ has one-sided error.  
    \end{claim}
    \begin{proof}
        If the input $s$ is in $\cP^r$ then $s=w^r$ for some string $w\in\cP$. In this case, \reptest accepts with probability $1$, and $\cT'$ accepts with the same probability as $\cT$. 
  Next, suppose $s$ is $\eps$-far from $\cP^r$. Consider the string $\hat{w}\in \Sigma^m$, where $\hat{w}_\mindex=\pl_{\rindex\in [r]} s[\rindex]_\mindex$ for all $j\in[m]$. Recall that the repetition code is defined as $\cC_{m,r}=\{s^{r}: s\in\Sigma^m\}$. Observe that the distance from $s$ to $\cC_{m,r}$ is equal to $\delta_H(s,\hat{w}^r)$. If this distance is at least $\frac\eps 2$ then, by \Cref{claim:reptest}, \reptest accepts with probability at most $\frac{1}{6}$, completing the proof.
        
        Now assume $\delta_H(s,\hat{w}^r)<\frac \eps 2.$ We will show that, in this case, the simulation of $\cT$ in \Cref{alg: online test} accepts with probability at most $\frac 1 6$. Specifically, we consider two failure events: let $E_1$ be the event that the simulation of $\cT$ feeds $\cT$ at least one query answer inconsistent with $\hat w$, and $E_2$ be the event that the simulation of $\cT$ accepts string $\hat w$. We will show that $\Pr[E_1]\leq \frac 1 {12}$ and $\Pr[E_2]\leq \frac 1 {12}$. Then a union bound over $E_1$ and $E_2$ completes the proof.
        
        Since $s$ is $\eps$-far from $\cP^r$ and $\delta_H(s,\hat{w}^r)<\frac \eps 2$, we get
        \begin{align*}
        \eps
        \leq\delta_H(s,\cP^r)
        \leq\delta_H(s,\hat{w}^r)+\delta_H(\hat{w}^r,\cP^r) 
        \leq\frac \eps 2+\delta_H(\hat{w},\cP),
        \end{align*}
        implying that $\delta_H(\hat{w},\cP)\geq\frac \eps 2$.
       Since $\cT$ is a tester for $\cP$ with amplified success probability, and $\hat{w}$ is $\frac\eps 2$-far from $\cP$, algorithm $\cT$ accepts with probability at most $\frac {1}{12}$ when run on input $\hat{w}$ and with proximity parameter $\frac\eps 2$; that is, $\Pr[E_2]\leq \frac 1 {12}$.
              
       Next we analyze $\Pr[E_1]$.
       Fix $\mindex\in [m]$, and let $a_\mindex$ denote the answer provided by $\cT'$ to query $\mindex$ in the simulation. If $\cT'$ does not reject while simulating this query, then
        \begin{align}
            \Pr[a_\mindex\neq \hat w_\mindex]&= \Pr_{\rindex_1,\dots, \rindex_d\in[r]}[s[\rindex_k]_\mindex=s[\rindex_{k'}]_\mindex \wedge s[\rindex_k]_\mindex\neq \hat w_\mindex \ \forall k,k'\in[d]]\nonumber\\
            &= \Pr_{\rindex_1,\dots,\rindex_d\in[r]}[s[\rindex_k]_\mindex=s[\rindex_1]_\mindex \ \forall k\in\{2,\dots,d\}\mid  s[\rindex_1]_\mindex\neq \hat w_\mindex]\cdot\Pr_{i_1\in[r]}[s[\rindex_1]_\mindex\neq \hat w_\mindex]\nonumber\\
            &\leq 2^{-d+1}\label{eq:query-simulation-bound}\\
            &=2(c_1\cdot q(m,\eps/2))^{-1}\nonumber. 
        \end{align}       
        The inequality in \eqref{eq:query-simulation-bound} follows from the following simple observation: once $s[\rindex_1]_\mindex$ has been sampled, if it is not the plurality, then for each $k\in\{2,\dots, d\}$, every subsequent $s[\rindex_k]_\mindex$ is equal to $s[\rindex_1]_\mindex$ with probability at most $\frac 12$. By our choice of $c_1$, and a union bound over all the $c_0\cdot q(m,\eps/2)$ queries made by $\cT$, the probability of $E_1$ is at most $\frac 1 {12}$.
        
        Since $\cT'$ can accept only if $E_1\cup E_2$ occurs, a union bound over $E_1$ and $E_2$ implies that $\cT'$ is indeed a tester for $\cP^r$, and has error probability at most $\frac{1}{6}$ when no adversary is present. 
    \end{proof}

Finally, we prove the lifting lemma.  
\begin{proof}[Proof of \Cref{lem: lifting}]
First, we argue that the probability $\cT'$ (\Cref{alg: online test})  queries a corrupted index is small. Recall that $n=m\cdot r$ is the length of the input to $\cT'$ and that $c_0\cdot q(m,\eps)$ is the query complexity of $\cT$ with proximity parameter $\eps$. Let $q'=q'(n,\eps)=\frac 8\eps+c_0\cdot q(m,\frac \eps 2)\log(c_1\cdot q(m,\frac \eps 2))$ be the number of queries made by $\cT'$. Then, there are at most $q'\rate$ corruptions at any point during the execution. Moreover, each query made by $\cT'$ is an index that is uniform over the set of corresponding indices in the $r$ different segments of $s$. It follows that each query made by $\cT'$ is corrupted with probability at most $\frac{q'\rate}{r}$. By a union bound over the $q'$ queries, 
    \begin{equation}
    \label{eq: corrupt query}
    \Pr[\text{$\cT'$ queries a corrupted index} ]\leq \frac{(q')^2\rate}{r}\leq\frac{1}{6}.
    \end{equation} 
    The upper bound follows from the hypothesis that $\rate\leq\frac{\delta\cdot r}{[q(m,\eps/2)\log q(m,\eps/2)]^2}$, that $q(m,\eps)=\Omega\big(\frac1\eps\big)$, and that $\delta$ is a sufficiently small constant. By a union bound, the probability that $\cT'$ errs in the presence of the adversary is at most the probability $\cT'$ errs when no adversary is present plus the probability $\cT'$ queries a corrupted index. By \Cref{claim: tester with no adversary} and \eqref{eq: corrupt query}, this probability is at most $\frac 1 3$, which completes the proof of \Cref{lem: lifting}. 
\end{proof}


\section{Randomness complexity separation}\label{sec:randomness-separation}

In this section, we prove our result separating the randomness complexity of testing in the online and offline models.   

\begin{theorem}[Randomness separation]
    \label{thm: rand sep}
    For all $\tau=\tau(n)$, there exist a property $\cP^{(\tau)}=\bigcup_{n\in\N} \cP^{(\tau)}_n$, where $\cP^{(\tau)}_n\subseteq [n]^n$, and an algorithm $\cT$ that takes $\eps\in(0,1)$ as an input and, for all $\eps\in(0,1)$,  is a one-sided error $\eps$-tester for $\cP^{(\tau)}$ that make $O(\frac \tau \eps)$ queries. Additionally, the following hold: 
    \begin{enumerate}
        \item\label{thm: rand sep item 1} For all $\eps,\alpha\in(0,1)$, tester $\cT$ is an $\alpha$-offline-erasure-resilient $\eps$-tester for $\cP^{(\tau)}$. Moreover, $\cT$ can be simulated using $O(\log \plen)$ random bits and no additional queries.
        \item\label{thm: rand sep item 2} For all $\eps\in(0,1)$ and $\rate\leq\left(\frac{0.01\eps\sqrt{\plen}}{\tau}\right)^2,$ tester $\cT$ is a $\rate$-online-erasure-resilient $\eps$-tester for $\cP^{(\tau)}$. 
        \item\label{thm: rand sep item 3} There exists a constant $\eps>0$ such that for all $\rate\in\N$ and $\tau\leq\frac{0.01\sqrt\plen}{\log\plen}$, every $\rate$-online-erasure-resilient $\eps$-tester for $\cP^{(\tau)}$ uses $\Omega\left(\frac{\tau\log(\rate+1)}{\log\tau}\right)$ random bits. Moreover, if the tester has one-sided error, then it uses $\Omega(\tau\log(\rate+1))$ random bits. 
    \end{enumerate}
\end{theorem}

\Cref{thm: rand sep} yields meaningful separations for all settings of $\tau$ between $\Omega(\log \plen)$ and $O\left(\frac{\sqrt\plen}{\log \plen}\right)$. As an example, the following corollary is derived by setting $\tau=\frac{0.01\eps\sqrt{\plen}}{\log\plen}$ and $\rate=1$. 
\begin{corollary}
\label{cor:rand sep cor 1}
There exists a property that can be $\eps$-tested (for every constant $\eps$) using $O\big(\frac{\sqrt{n}}{\log n}\big)$ queries in the presence of either an online or 
offline adversary. The offline-erasure-resilient tester uses $O(\log n)$ random bits, but there exists a constant $\eps$ such that every $1$-online-erasure-resilient $\eps$-tester, requires $\Omega\big(\frac{\sqrt{n}}{\polylog n}\big)$ random bits to succeed with constant probability. Moreover, if the $1$-online-erasure-resilient tester has one-sided error then it requires $\Omega(\sqrt n)$ random bits.
\end{corollary}

 To prove our randomness lower bound, we present a general reduction from randomness-efficient testers in the online-erasure model to query-efficient testers in the standard model. 
 \begin{lemma}[Generalization of \Cref{lem: rand-query reduction intro}]\label{lem: rand-query reduction}
        For all  $r,\rate\in\N$ and properties $\cP$, if
    $\cP$ is $\eps$-testable in the presence of a $\rate$-online-erasure adversary using only $r$ bits of randomness then $\cP$ is  $\eps$-testable in the standard model using at most $\frac r{\log(\rate+1)}$ queries. 
\end{lemma}

\begin{proof}
    Let $\cT$ be a $\rate$-online-erasure-resilient $\eps$-tester for $\cP$ that uses at most $r$ random bits. We construct an adversary that allows $\cT$ to make at most $\frac{r}{\log(\rate+1)}$ non-erased queries. Consider the following random seed elimination adversary  $\cA$.     Given an input $x$, the description of algorithm $\cT$, and query-answer history of $\cT$ on input $x$, the adversary simulates the next query of $\cT$ on every random seed $s\in\zo^r$ that is consistent with the query-answer history of $\cT$ thus far and erases the $\rate$ most queried indices of $x$. Intuitively, these are the indices that are most likely to be queried by the tester at this step, based on the random seeds that are consistent with the current query-answer history. Each index that the oracle decides not to erase at this step can appear in at most $\frac 1{\rate+1}$ fraction of currently relevant random seeds.  If $\cT$ queries such an index, all random seeds that would have lead to a different query can be eliminated. Thus, each time $\cT$ makes a non-erased query, at most a $\frac 1{\rate+1}$ fraction of the relevant random seeds remain. Consequently, $\cT$ can make at most $\frac r{\log(\rate+1)}$ non-erased queries before $\cA$ can exactly determine the random seed being used. Once $\cA$ determines the random seed, it can exactly predict the queries of $\cT$ and erase all of them. 

    Next, we use the adversary strategy $\cA$ to construct an $\frac r{\log(\rate+1)}$-query tester for $\cP$. Let $\cT^*$ be defined as follows. Draw a random seed $s^*\in\zo^r$ and simulate $\cT$ with random seed $s^*$. After making each query, $\cT^*$ simulates $\cA$ with the query-answer history of $\cT$. If $\cA$ erases the next query of $\cT$, then $\cT^*$ answers the query with $\perp$ and continues the simulation (the algorithm $\cT^*$ need not make any queries in this case). Otherwise, if $\cA$ does not erase the next query of $\cT$, then $\cT^*$ queries $x$, provides $\cT$ with the answer, and continues the simulation. Since $\cT$ makes at most $\frac r{\log(\rate+1)}$ non-erased queries and successfully tests $\cP$, the algorithm $\cT^*$ makes at most $\frac r{\log(\rate+1)}$ queries and also successfully tests $\cP$.  
\end{proof}

\subsection{The \texorpdfstring{$\tau$}{tau}-Distinct-Elements property}
The property testing problem we use to prove \Cref{thm: rand sep} is a version of support size approximation. In support size approximation, one is given sample access to a distribution $\cD$ over the domain $[n]$ and asked to approximate $|\supp(\cD)|$ up to additive error $\eps$. We study the analogous testing problem over strings: Given parameters $\tau\in\N$ and $\eps\in(0,1)$, and a string $x$ of length $\plen$, we wish to $\eps$-test whether $|\{x_1,x_2,\dots,x_\plen\}|\leq\tau$, that is, to determine whether $x$ has at most $\tau$ distinct elements or $x$ is $\eps$-far from every $x'$ (of length $\plen$) with at most $\tau$ distinct elements. 

\begin{definition}[$\tau$-Distinct-Elements]
    \label{def: de}
    For all $n\in\N$ and $\tau=\tau(n)$, let $\cP^{(\tau)}_n$ be the set of strings in $[n]^n$ with at most $\tau$ distinct symbols. The property  {\em$\tau$-Distinct-Elements} is defined as $\cP^{(\tau)}=\bigcup_{n\in\N} \cP^{(\tau)}_n$.
\end{definition}

The following tester for $\tau$-Distinct-Elements appeared in \cite{GoldreichR23}. We include 
the analysis of its guarantees here for completeness.

\begin{fact}[A tester for $\tau$-Distinct-Elements]
    \label{fact:de-tester}
      For all $\eps\in(0,1)$, there exists a one-sided error tester for $\cP^{(\tau)}$ that uses $O(\frac\tau\eps)$ samples and fails with probability at most $\frac 1 {10}$. 
\end{fact}
\begin{proof}
    We first make the following observation: if $x$ is $\eps$-far from $\cP^{(\tau)}$, then every collection of $\tau$ distinct elements occupies at most a $1-\eps$ fraction of the indices in $x$. Suppose we sample uniformly random elements from $x$ until we have seen $\tau+1$ distinct elements. Notice that while we have seen at most $\tau$ distinct elements, our next sample is a new distinct element with probability at least $\eps$ (the at most $\tau$ distinct elements we have seen can have probability mass at most $1-\eps$). Let $X_i$ be the number of samples to get the $i$-th distinct element given we have seen $i-1$ distinct elements, and set $X=\sum_{i=1}^{\tau+1} X_i$. Then $\Ex[X]\leq (\tau+1)/\eps$, and hence, by Markov's inequality, the probability we need more than $3(\tau+1)/\eps$ samples is at most $\frac 1 3$. This analysis inspires the following simple tester:

    {\em Sample $s=3(\tau+1)/\eps$ elements from $x$ and accept if and only if there are at most $\tau$ distinct symbols from $[n]$ in the sample (not counting the erasure symbol $\perp$).}
    
    Repeating the tester constantly many times suffices to make the failure probability smaller than $\frac 1{10}$.  
\end{proof}

Recall that \Cref{lem: rand-query reduction} states that to prove a lower bound on the randomness complexity of testing in the online model, it suffices to prove a lower bound on the query complexity of testing in the standard model. In \Cref{lemma:ds-lower-bound}, we show a lower bound on the query complexity of testing $\tau$-Distinct-Elements. Our lower bound is a direct corollary of the lower bound of \cite{ValiantV10}. 

\begin{lemma}[Query lower bound for testing $\tau$-Distinct-Elements]
    \label{lemma:ds-lower-bound}
    There exists $\eps\in(0,1)$ such that for all input lengths $n$, if $\tau\leq\frac{\sqrt\plen}{\log\plen}$, then 
    the query complexity of $\eps$-testing $\cP^{(\tau)}$ is $\Omega(\frac\tau{\log\tau})$. Moreover, the query complexity of $\eps$-testing $\cP^{(\tau)}$ with one-sided error is $\Omega(\tau)$.
\end{lemma}
\begin{proof}
    We use the following fact to argue that we can assume the tester is sample-based with no loss of generality. 

    \begin{fact}[Theorem 6.1 \cite{GoldreichR15}]
    \label{fact: GR15}
    A property of strings is \emph{symmetric} if it is closed under permutations of the string indices (that is, under reordering of the characters in the string). Let $\cP$ be a symmetric property of strings that is $\eps$-testable with query complexity $q$. Then there exists a sample-based tester for $\cP$ with sample complexity $O(q)$.
    \end{fact}
    Since $\cP^{(\tau)}$ is symmetric, we can without loss of generality consider sample-based testers. First, we prove the lower bound for two-sided error testers. The essence of our proof is a reduction from the hard instances used in \cite{ValiantV10} for estimating the support size of a distribution. While we cannot use the lower bound of \cite{ValiantV10} as a black box, we can nonetheless leverage their hard instances. Indeed, combining Theorem 1, Fact 5, and the remarks regarding the support size of their hard distributions (following Definition 12) from \cite{ValiantV10}, we obtain the following immediate corollary.

    \begin{fact}[Support size approximation lower bound  \cite{ValiantV10}]
        \label{fact:VV10-thm1}
        Let $\alpha>0$ be a sufficiently small constant. There exist distributions $\cD^+$ and $\cD^-$ over $[m]$ such that
        \begin{enumerate}
            \item Every element in the support of $\cD^+$ or $\cD^-$ has probability mass at least $\frac1m$.
            \item $|\supp(\cD^+)|\leq m\left(\frac12+\alpha\right)$, and $\cD^-$ is $\left(\frac12-2\alpha\right)$-far from every distribution with support size at most $m\left(\frac12+\alpha\right)$.
        \end{enumerate}
        Moreover, every algorithm that successfully distinguishes $\cD^+$ from $\cD^-$ with constant probability requires $\Omega\left(\frac{m}{\log m}\right)$ samples.         
    \end{fact}

    While the distributions in \cite{ValiantV10} are defined over $\R$, they all have support size at most $m$. Thus, we can, without loss of generality, consider distributions over $[m]$. Indeed, given a distribution $\cD$ over $\R$, the distribution $\cD^\pi$ obtained by choosing a uniformly random permutation $\pi$ of $[m]$, and then mapping the $i$-th element in the support of $\cD$ to $\pi(i)$, has domain $[m]$ and the same support size as $\cD$. Moreover, if two distributions $\cD_0$ and $\cD_1$ each have support size at most $m$, then for $b\in\zo$, the distribution $\cD^\pi_b$ can be generated by first, sampling a random permutation $\pi$ of $[m]$, second, sampling from $\cD_b$, and third, sending each distinct sample to the first available element of $\pi$. Thus, any algorithm which distinguishes $\cD^\pi_1$ and $\cD^\pi_0$ can be used to distinguish $\cD_0$ and $\cD_1$ with the same sample complexity.
    
    To apply \Cref{fact:VV10-thm1} in our setting, we note that whenever $m\leq \frac{\sqrt\plen}{\log\plen}$, there exists $x\in[m]^\plen$ such that the distribution obtained by sampling a uniform index in $x$ has total variation distance at most $\frac m\plen=\frac{1}{\sqrt\plen\log\plen}$ from $\cD$. Thus, for  sufficiently large $\plen$, there exist strings $x^+,x^-\in[m]^\plen$ such that $x^+$ has at most $m(\frac12+\alpha)$ distinct symbols, and the distributions generated by sampling from $x^+$ and $x^-$ have total variation distance $\frac{1}{\sqrt\plen\log\plen}$ from $\cD^+$ and $\cD^-$, respectively. Let $\tau=m(\frac12+\alpha)$ and $\eps=\frac12-2\alpha-o(1)$, then $x^+\in\cP^{(\tau)}$ and $x^-$ is $\eps$-far from $\cP^{(\tau)}$. By \Cref{fact:VV10-thm1}, every algorithm for distinguishing $x^+$ from $x^-$ uses $\Omega\left(\frac{m}{\log m}\right)=\Omega\left(\frac{\tau}{\log\tau}\right)$ samples. By \Cref{fact: GR15}, every tester for $\cP^{(\tau)}$ has query complexity $\Omega\left(\frac{\tau}{\log\tau}\right)$. 

    To prove the one-sided error lower bound, it suffices to note that every one-sided error algorithm must accept if its sample contains at most $\tau$ symbols. Thus, every one-sided error tester requires $\Omega(\tau)$ samples.
    This completes the proof of \Cref{lemma:ds-lower-bound}.
\end{proof}

\subsection{Proof of the randomness separation (\protect{\Cref{thm: rand sep}})}

To complete the proof of the theorem, we argue that the tester given by \Cref{fact:de-tester} is a valid tester for $\tau$-Distinct-Elements in both the online and offline erasure models. To prove the lower bound on the randomness complexity of testing $\tau$-Distinct-Elements in the online model, we combine the query lower bound of \Cref{lemma:ds-lower-bound} with the reduction stated in \Cref{lem: rand-query reduction}. To prove the upper bound on the randomness complexity of testing $\tau$-Distinct-Elements in the offline model, we argue that the offline tester can be simulated with few bits of randomness at a small cost in error probability.

We first prove the part of the theorem regarding testing $\tau$-Distinct-Elements in the offline erasure model.

\begin{proof}[Proof of \protect{\Cref{thm: rand sep}}, \protect{\Cref{thm: rand sep item 1}}]
     First, for all $\tau\in\N$, we construct an offline-erasure-resilient tester for $\tau$-Distinct-Elements, and second, we argue that the tester can be simulated with $O(\log \plen)$ random bits and no additional queries. Let $\cT$ be the tester given by \Cref{fact:de-tester}. Recall that $\cT$ accepts if and only if the number of nonerased symbols in the sample is at most $\tau$. Below, we argue that $\cT$ is an $\alpha$-offline-erasure-resilient $\eps$-tester. 
     
     Let $x$ be an $\alpha$-erased string in $[n]^n$. If some completion of $x$ is in $\cP^{(\tau)}$, then the number of nonerased symbols in the sample drawn by $\cT$ is at most $\tau$, so the algorithm accepts. Next we argue that $\cT$ accepts with probability at most $\frac 1 4$ when every completion of $x$ is $\eps$-far from $\cP^{(\tau)}$. Let $x^{\perp}$ denote the erased portion of $x$, and $x^{\top}$ denote the nonerased portion of $x$. Notice that $|x^{\top}|=(1-\alpha)n$, and thus if $x^\top$ is $\frac{\eps}{1-\alpha}$-close to $\cP^{(\tau)}$, then there is a completion of $x$ that is $\eps$-close to $\cP^{(\tau)}$ (change $\eps n$ symbols in $x^{\top}$ and fill out $x^{\perp}$ with an arbitrary symbol from $x^{\top}$). 
     Thus, if every completion of $x$ is $\eps$-far from $\cP^{(\tau)}$, then $x^{\top}$ is $\frac{\eps}{1-\alpha}$-far from $\cP^{(\tau)}$. 
     By the proof of \Cref{fact:de-tester}, a set of $O\big(\frac{\tau(1-\alpha)}{\eps}\big)$ uniform and independent samples from $x^{\top}$ contains at least $\tau+1$ distinct symbols with probability at least $\frac 9 {10}$. Moreover, a sample of $O\big(\frac{\tau}{\eps}\big)$ elements from $x$ contains $\Omega\big(\frac{\tau(1-\alpha)}{\eps}\big)$ such samples from $x^{\top}$ with probability at least $\frac 9 {10}$. By the union bound, $\cT$ accepts $x$ with probability at most $\frac 1 4$.  

    It remains to show that $\cT$ can be simulated by a tester $\cT'$ that uses $O(\log \plen)$ random bits and has the same query complexity as $\cT$. \Cref{fact: GS10} states that any randomized oracle machine that solves a promise problem can be simulated using $\log\plen+\log\log|\Sigma|+O(1)$ random bits, where $|\Sigma|$ is the size of the alphabet. This fact was originally proven by Goldreich and Sheffet \cite{GS10} and is stated as an exercise in \cite{Go-book}. We use the statement from the exercise since it is more convenient for our setting.  

    \begin{fact}[Exercise 1.21 \cite{Go-book}]
    \label{fact: GS10}
    Fix $\plen\in\N$ and let $\Pi$ be a promise problem regarding strings in $\Sigma^\plen$. Suppose that $\cM$ is a randomized $q$-query oracle machine that solves $\Pi$ with error probability at most $\frac 14$. Then there exists a randomized $q$-query oracle machine $\cM'$ that solves $\Pi$ with error probability at most $\frac 13$ and $\log \plen+\log\log |\Sigma|+O(1)$ random bits. 
    \end{fact}

    Let $\Pi$ be the promise problem defined by $\alpha$-offline-erasure-resilient $\eps$-testing $\cP^{(\tau)}$. By the first part of the proof, there exists a randomized oracle machine $\cT$ that solves $\Pi$ with error probability at most $\frac 14$. 
   Applying \Cref{fact: GS10} directly, we see that $\cT$ can be simulated by a tester $\cT'$ using $\log \plen+\log\log|\Sigma|+O(1)$ random bits; since $\Sigma=[n]$, tester $\cT'$ uses $O(\log \plen)$ random bits, and has error probability at most $\frac 13$. To complete the proof, it suffices to note that $\cT'$ has the same query complexity as $\cT$.
\end{proof}

Next, we prove the part of the theorem regarding testing $\tau$-Distinct-Elements in the online erasure model.

\begin{proof}[Proof of \protect{\Cref{thm: rand sep}}, Items \ref{thm: rand sep item 2} and \ref{thm: rand sep item 3}]
    We will show that for all $\tau\in\N$, the tester given by \Cref{fact:de-tester} is online-erasure-resilient. Let $\cT$ be the tester for $\tau$-Distinct-Elements given by \Cref{fact:de-tester}. Recall that $\cT$ accepts if and only if the number of nonerased symbols in its sample is at most $\tau$. We show that $\cT$ can test in the presence of a $\rate$-online-erasure-oracle, by arguing that with high constant probability, none of the queries made by $\cT$ are erased. Notice that with proximity parameter $\eps$, the tester makes $O\left(\frac\tau\eps\right)$ queries. Thus, there are $O\left(\frac{\rate\cdot\tau}{\eps}\right)$ erasures at every point during the execution of $\cT$. Since $\cT$ makes uniform and independent queries, the probability that any particular query is erased is $O\left(\frac{\rate\cdot \tau}{\plen\eps}\right)$. Since $\cT$ makes $O\left(\frac{\tau}{\eps}\right)$ queries, we can apply the union bound to see that some query is erased with probability at most $O\left(\frac{\rate \cdot\tau^2}{\plen\eps^2}\right)$. By the hypothesis that $\rate\leq \left(\frac{.01\eps\sqrt \plen}{\tau}\right)^2$, we have $O\left(\frac{\rate \cdot\tau^2}{\plen\eps^2}\right)\leq.01$ -- that is, $\cT$ sees an erasure with probability at most $.01$. By \Cref{fact:de-tester}, the algorithm $\cT$ is a $\rate$-online-erasure-resilient $\eps$-tester for $\cP^{(\tau)}$. Moreover, since an erasure cannot increase number of distinct nonerased symbols in $x$, tester $\cT$ has one-sided error. This completes the proof of \Cref{thm: rand sep item 2}.
    
    To see why \Cref{thm: rand sep item 3} holds, recall that by \Cref{lemma:ds-lower-bound}, there exists $\eps\in(0,1)$ such that for all $\plen,\tau\in\N$ with $\tau\leq \frac{\sqrt \plen}{\log \plen}$, every $\eps$-tester for $\cP^{(\tau)}$ has query complexity $\Omega\left(\frac\tau{\log\tau}\right)$. Similarly every tester with one-sided error has query complexity $\Omega(\tau)$. Thus, applying \Cref{lem: rand-query reduction}, the reduction from query complexity to randomness complexity, suffices to complete the proof. 
\end{proof}

\section*{Acknowledgment}
The authors thank Uri Meir for fruitful discussions regarding the online model. 
\bibliographystyle{alpha}  
\bibliography{ref} 

\end{document}